\documentclass[11pt]{article}
\usepackage{fullpage}

\usepackage{times}
\usepackage{comment,amsfonts,amssymb,amsmath,amsthm,graphicx,algorithm,algorithmic}
\newcommand{\commentout}[1]{}

\ifx\pdftexversion\undefined
\usepackage[colorlinks,linkcolor=black,filecolor=black,citecolor=black,urlco
lor=black,pdfstartview=FitH]{hyperref}
\else
\usepackage[colorlinks,linkcolor=blue,filecolor=blue,citecolor=blue,urlcolor
=blue,pdfstartview=FitH]{hyperref}
\fi

\newcommand{\alert}[1]{\textbf{\color{red}
[[[#1]]]}\marginpar{\textbf{\color{red}**}}\typeout{ALERT:
\the\inputlineno: #1}}

\def\MathF{\hbox{\rm I\kern-2pt F}}
\def\MathP{\hbox{\rm I\kern-2pt P}}
\def\MathR{\hbox{\rm I\kern-2pt R}}
\def\MathZ{\hbox{\sf Z\kern-4pt Z}}
\def\MathN{\hbox{\rm I\kern-2pt I\kern-3.1pt N}}
\def\MathC{\hbox{\rm \kern0.7pt\raise0.8pt\hbox{\footnotesize I}
\kern-4.2pt C}}
\def\MathQ{\hbox{\rm I\kern-6pt Q}}

\newcommand{\final}{{\rm final}}
\newcommand{\mmod}{{\rm mod } }

\newcommand{\diam}{{\rm diam}}

\newcommand{\N}{\mathbb{N}}

\newcommand{\R}{\mathbb{R}}

\newcommand{\mommit}[1]{}
\newcommand{\namedref}[2]{\hyperref[#2]{#1~\ref*{#2}}}
\newcommand{\sectionref}[1]{\namedref{Section}{#1}}

\newcommand{\theoremref}[1]{\namedref{Theorem}{#1}}

\newcommand{\algref}[1]{\namedref{Algorithm}{#1}}
\newcommand{\claimref}[1]{\namedref{Claim}{#1}}
\newcommand{\lemmaref}[1]{\namedref{Lemma}{#1}}

\newcommand{\corollaryref}[1]{\namedref{Corollary}{#1}}

\newtheorem{theorem}{Theorem}
\newtheorem{lemma}{Lemma}
\newtheorem{corollary}[lemma]{Corollary}

\newtheorem{claim}[lemma]{Claim}

\usepackage{pdfsync}
\usepackage{authblk}

\title{Near Isometric Terminal Embeddings for Doubling Metrics}

\author[1]{Michael Elkin\thanks{This research was supported by the ISF grant No. (724/15).}}
\author[1]{Ofer Neiman\thanks{Supported in part by the ISF grant 1817/17 and BSF grant 2015813.}}

\affil[1]{Department of Computer Science, Ben-Gurion University of the Negev,
Beer-Sheva, Israel. Email: \texttt{\{elkinm,neimano\}@cs.bgu.ac.il}}

\usepackage{pdfsync}

\begin{document}

\maketitle

\begin{abstract}
Given a metric space $(X,d)$, a set of terminals $K\subseteq X$, and a parameter $t\ge 1$, we consider metric structures (e.g., spanners, distance oracles, embedding into normed spaces) that preserve distances for all pairs in $K\times X$ up to a factor of $t$, and have small size (e.g. number of edges for spanners, dimension for embeddings). While such terminal (aka source-wise) metric structures are known to exist in several settings, no terminal spanner or embedding with distortion
close to 1, i.e., $t=1+\epsilon$ for some small $0<\epsilon<1$, is currently known.

Here we devise such terminal metric structures for {\em doubling} metrics, and show that essentially any metric structure with distortion $1+\epsilon$ and size $s(|X|)$ has its terminal counterpart, with distortion $1+O(\epsilon)$ and size $s(|K|)+1$. In particular, for any doubling metric on $n$ points, a set of $k=o(n)$ terminals, and constant $0<\epsilon<1$, there exists
\begin{itemize}
\item A spanner with stretch $1+\epsilon$ for pairs in $K\times X$, with $n+o(n)$ edges.
\item A labeling scheme with stretch $1+\epsilon$ for pairs in $K\times X$, with label size $\approx \log k$.
\item An embedding into $\ell_\infty^d$ with distortion $1+\epsilon$ for pairs in $K\times X$, where $d=O(\log k)$.
\end{itemize}
Moreover, surprisingly, the last two results apply if only $K$ is a doubling metric, while $X$ can be arbitrary.
\end{abstract}

\section{Introduction}

The area of {\em low-distortion embeddings} studies how well different metric spaces can be approximated by simpler, or more structured, metric spaces. Fundamental results in this realm include Bourgain's and Matousek's embeddings of general metrics into high-dimensional Euclidean  and $\ell_\infty$ spaces \cite{B85,M96}, respectively,  Gupta et al.'s \cite{GKL03} embeddings of doubling metrics into normed spaces, and constructions of distance oracles and spanners for doubling metrics \cite{HM06,GGN06}. Linial et al. \cite{LLR95} and Bartal \cite{B96} demonstrated that low-distortion embeddings have numerous applications in Theoretical Computer Science.

All these embeddings \cite{B85,M96,GKL03} have inherent unavoidable dependencies in the total number of points $n$ in both the distortion and in the dimension of the target space. In scenarios in which we have a metric space $(X,d)$, and a subset $K \subseteq X$ of important points, aka {\em terminals}, the current authors and Filtser \cite{EFN17}  demonstrated that one can devise {\em terminal embeddings}, i.e., embeddings that provide guarantees on the distortion of all pairs that involve a terminal in $K$, and whose guarantees on the distortion and the dimension depend on $k = |K|$, as opposed to the dependencies on $n$ in the classical embeddings.
Specifically, it is shown in \cite{EFN17} that essentially any known metric embedding into a normed space can be transformed via a general transformation into a terminal embedding, while incurring only a constant overhead in distortion.

This constant overhead does not constitute a problem when the distortion of the original embedding is $O(\log n)$, as is the case for Bourgain's embedding. However, for the important family of embeddings of {\em doubling} metrics \cite{A83,GKL03} the distortion in some cases is just $1+ \epsilon$, for an arbitrarily small $\epsilon > 0$. (The dimension grows with $1/\epsilon$.) This is also the case in the constructions of spanners and distance oracles for these metrics, due to \cite{T04,GGN06,HM06}. Using the general transformation of \cite{EFN17} on them results in stretch $c$, for some constant $c \ge 1 + \sqrt{2}$, making the resulting embeddings and spanners far less appealing.

A metric $(X,d)$ has doubling constant $\lambda$ if any ball of radius $2R$ in the metric (for any $R>0$) can be covered by at most $\lambda$ radius-$R$ balls. The parameter $\log_2\lambda$ is called also the {\em doubling dimension} of the metric $(X,d)$. A family of metrics is called {\em doubling} if the doubling dimension of each family member is constant.

Doubling metrics constitute a useful far-reaching generalization of Euclidean low-dimensional metrics. They have been extensively studied, see \cite{A83,GKL03,CG06,HM06,GGN06,CGMZ06,GR08,CLNS15,ES15,G15,N16} and the references therein. Interestingly, these studies of doubling metrics have often produced improved bounds for low-dimensional Euclidean metrics as well. This was the case, e.g., for dynamic spanners for doubling and low-dimensional Euclidean metrics \cite{GR08}, spanners with low diameter, degree and weight \cite{ES15}, and fault-tolerant spanners \cite{CLNS15}.

In the current paper we devise a suit of terminal embeddings and metric structures, such as spanners, distance oracles and distance labeling schemes (see \sectionref{sec:prem} for definitions), for doubling metrics with distortion $1+\epsilon$, for an arbitrarily small $\epsilon > 0$. In particular, Gupta et al. \cite{GKL03} devised an embedding of metrics with doubling constant $\lambda$ into $\ell_\infty$ with distortion $1+\epsilon$ and dimension $\log n \cdot \lambda^{\log 1/\epsilon + O(1)}$.  Our terminal embedding of doubling metrics into $\ell_\infty$ has {\em the same distortion}, but the dimension is
$\log k \cdot \lambda^{\log 1/\epsilon + O(1)}$, i.e., the dependency on $n$ is replaced by (essentially) the same  dependency on $k$.

Johnson and Lindenstrauss \cite{JL84} showed that any Euclidean metric can be embedded into an $O({{\log n} \over {\epsilon^2}})$-dimensional  Euclidean one, with distortion $1+ \epsilon$. While we are not able to provide a general terminal counterpart of this fundamental result, we do so in the important special case of doubling metrics. Specifically, we show that an Euclidean (possibly high-dimensional\footnote{By ``high-dimensional" we mean here typically dimension $\log n$ or greater.}) point set with doubling constant $\lambda$ admits a terminal embedding with distortion $1+\epsilon$ into an Euclidean space with dimension $O((\log k + \log \lambda \cdot \log 1/\epsilon) /\epsilon^2)$.

Har-Peled and Mendel \cite{HM06}, following \cite{T04}, and extending previous classical results about low-dimensional Euclidean spanners (see, e.g., \cite{ADDJS93,CDNS92,DHN93,NS07}), showed that  for any $n$-point metric with doubling constant $\lambda$ and $\epsilon > 0$, there exists a $(1+\epsilon)$-spanner with $n \cdot \lambda^{O(\log 1/\epsilon)}$ edges. Note that when $\epsilon$ is very small, the coefficient of $n$ may be pretty large even in Euclidean two-dimensional space. We devise a terminal $(1+\epsilon)$-spanner for doubling metrics with $n + k \cdot \lambda^{O(\log 1/\epsilon)}$ edges. In other words, when the number of terminals $k$ is much smaller than $n$, the number of edges is just $n+o(n)$, as opposed to $n$ multiplied by a large constant.
(Note, however, that the distortion that our spanner provides is for pairs in $K \times X$,  as opposed to $X \times X$.)
To the best of our knowledge, no such terminal spanners are known even for two-dimensional Euclidean point sets.

We also provide analogous terminal counterparts of Har-Peled and Mendel's distance oracles \cite{HM06}, and Slivkins' distance labeling schemes \cite{S07}.

In addition, we study the setting in which the set of terminals $K$ induces a doubling metric, while the entire point set $X$ is a general (as opposed to doubling) metric. Surprisingly, we show that our terminal distance labeling and also embedding of doubling metrics into $\ell_\infty$ apply in this far more general scenario as well, with the same stretch $1+\epsilon$, and the same size/dimension as when $X$ is a doubling metric. 
We also devise terminal spanners and terminal distance oracles for this more general scenario that $K$ is doubling, while $X$ is a general metric.

{\bf Related Work:}
There has been several works which devised metric structures for partial subsets. Already \cite{CE05} considered distance preservers for a designated set of pairs. In \cite{CGK13,P14,K15} pairwise spanners for general metrics were studied, and in particular terminal spanners. Recently \cite{AB18} introduced reachability preservers from a given set of sources.

Interestingly, lately we realized that the general transformation from \cite{EFN17} can also be easily extended to produce terminal embeddings that apply to this general scenario (that points of $X \setminus K$ lie in a general metric, while points of $K$ lie in a special metric). However, as was mentioned above, that transformation increases the stretch by at least a constant factor, and is thus incapable of producing terminal embeddings with stretch $1+\epsilon$.

The only known to us terminal metric structure with distortion $1+\epsilon$ is a prioritized distance labeling scheme for graphs that exclude a fixed minor, due to the current authors and Filtser \cite{EFN15}. In the current paper we provide the first near-isometric (i.e., having stretch $1+\epsilon$) terminal {\em spanners and embeddings}.

\subsection{Technical Overview}

The naive approach for building a terminal spanner for a given metric space $(X,d)$, is to apply a known construction on the set of terminals $K$, and extend the spanner to $X\setminus K$ by adding an edge from each point in $X\setminus K$ to its nearest terminal. (The same approach can be used for distance oracles/labeling and embeddings.) This is essentially the approach taken by \cite{EFN17} (albeit in a much more general setting). Unfortunately, such a construction cannot provide small $1+\epsilon$ stretch (it can be easily checked that it may give stretch at least 3). We need several ideas in order to provide small stretch.


First, we use the well known property of doubling metrics, that balls contain bounded size nets (see \sectionref{sec:prem} for definitions). We construct nets in all relevant distance scales, and enrich $K$ by a set $Y\supseteq K$ of net points. The points of $Y$ are those net points that are, to a certain extent, close to $K$, depending on their distance scale. Then we apply a black-box construction of a spanner on the set $Y$. Finally, we extend the spanner to every $x\in X\setminus Y$, by adding a single edge from $x$ : either to the nearest terminal, or to a single net point $y\in Y$. The set $Y$ is carefully chosen so that each non-terminal $x\in X\setminus K$, either has a close-by terminal that "takes care" of it, and otherwise there is a net point $y\in Y$ sufficiently close to $x$ so that $x$ will have good stretch going via $y$.

One issue to notice is that even though $Y$ is larger than $K$, it is still $|Y|=O(|K|)$ (at least for constant $\epsilon,\lambda$). So we can have many points in $X\setminus Y$ that do not have a representative $y\in Y$. The main technical part of the paper is devoted to proving that the particular choice of $Y$ guarantees low stretch for any pair $(x,v)\in X\times K$, even when $x$ has no representative $y\in Y$, by using the path through the nearest terminal to $x$.

It is instrumental to think of the set $Y$ as an {\em "enriched"} terminal set. This idea of enriching the terminal set $K$ with additional points may be useful in other settings as well.

In the setting when only $K$ is doubling, our construction of terminal spanners (and also distance oracles/labeling schemes) is done by adding {\em multiple edges} from each $x\in X\setminus K$ to nearby terminals that constitute a net. This approach can not work, however, for embeddings into normed spaces. A certain type of embedding (such as the embedding of doubling metrics into $\ell_\infty$) can be used in a non-black-box manner, and we show how to incorporate the points of $X\setminus K$ into the embedding for $K$, without increasing the dimension.

\section{Preliminaries}\label{sec:prem}

\subsection{Embeddings, Spanners and Distance Oracles/Labeling Scheme}
Let $(X,d)$ be a finite metric space. For a target metric $(Z,d_Z)$, an {\em embedding} is a map $f:X\to Z$, and the {\em distortion} of $f$ is the minimal $\alpha$ (in fact, it is the infimum), such that there exists a constant $c$ that for all $x,y\in X$
\begin{equation}\label{eq:dist}
d(x,y)\le c\cdot d_Z(x,y)\le \alpha\cdot d(x,y)~.
\end{equation}
When $Z$ is the shortest path metric of a graph $H$ and $c=1$, we say that $H$ is an $\alpha$-{\em spanner} of $(X,d)$.
Given a set of terminals $K\subseteq X$, a {\em terminal} embedding guarantees \eqref{eq:dist} only for pairs in $K\times X$.

An {\em approximate distance oracle} is a data structure that can report a multiplicative approximation of $d(x,y)$, for all $x,y\in X$. For $K\subseteq X$, it is a {\em terminal} distance oracle if it can report only pairs in $K\times X$. The relevant parameters of an oracle are: its size (we measure the size in machine words), query time, and stretch factor (and to some extent, also the preprocessing time required to compute it). If one can distribute the data structure by storing a short label $L(x)$ at each vertex $x\in X$, and compute the approximation to $d(x,y)$ from $L(x)$ and $L(y)$ alone, this is called a {\em distance labeling scheme}.

For $x\in X$ and $r>0$, let $B(x,r)=\{y\in X~:~ d(x,y)\le r\}$ be a closed ball. The doubling constant of $X$, denoted $\lambda$, is the minimal integer such that for every $r>0$, every ball of radius $2r$ can be covered by $\lambda$ balls of radius $r$.

\subsection{Terminal Nets}
For $r>0$, an {\em $r$-net} is a set $N\subseteq X$ satisfying the following:
\begin{enumerate}
\item For all $u,v\in N$, $d(u,v)>r$, and
\item for each $x\in X$, there exists $u\in N$ with $d(x,u)\le r$.
\end{enumerate}
The following claim is obtained by iteratively applying the definition of doubling constant.
\begin{claim}[\cite{GKL03}]\label{prop:double}
Fix any $q,r>0$, and let $N$ be an $r$-net. For any $x\in X$ we have that
$$
|B(x,q)\cap N|\le \lambda^{\log\lceil 2q/r\rceil}~.
$$
\end{claim}
It is well-known that a greedy algorithm that iteratively picks an arbitrary point $u\in X$ to be in $N$, and removes every point within distance $r$ of $u$, will create an $r$-net. Given a set of terminals $K\subseteq X$, we say that the greedy algorithm constructs a {\em terminal $r$-net}, if it prefers to take points from $K$ until it is exhausted, and only then picks other points to $N$.
We also observe that given a terminal $2r$-net $N$, one may choose a terminal $r$-net $N'$ that contains every {\em terminal} of $N$ (by greedily picking to $N'$ the terminals of $N$ first -- note that $N'$ is not guaranteed to contain all points of $N$, just the terminals).

\subsection{Extendable Metric Structure}
Given a metric $(X,d)$, we denote by $\hat{d}$ the distance function of some metric structure on it.
We say that a family of structures is {\em extendable}, if the structure on a subset $Y\subseteq X$ can be extended to the entire $X$ (so that $\hat{d}$ remains the same for pairs in $Y$), by {\em hanging} each $x\in X\setminus Y$ on some $u=u(x)\in Y$ and having that:
\begin{enumerate}
\item $\hat{d}(x,u)=d(x,u)$.
\item For any $v\in Y$, $\max\{d(x,u),\hat{d}(u,v)\}\le\hat{d}(x,v)\le d(x,u)+\hat{d}(u,v)$.
\label{item:2}
\end{enumerate}
We argue that essentially all known structures are extendable. For each $x\in X\setminus Y$, let $u=u(x)\in Y$ be the point onto which $x$ is hanged.
\begin{itemize}
\item {\bf Spanners}. If the structure is a spanner on $Y$, then the extension for each $x$ is done by adding the edge $\{x,u\}$ with weight $d(x,u)$. For any $v\in Y$, we indeed have that $\hat{d}(x,v)=d(x,u)+\hat{d}(u,v)$, satisfying both requirements.

\item {\bf Distance labeling}. For a distance labeling (or oracle), $x$ stores the label of $u$ and also $d(x,u)$. For a query on $(x,v)$ where $v\in Y$, return $\hat{d}(x,v)=d(x,u)+\hat{d}(u,v)$.

\item {\bf Embeddings}. If the structure is an embedding $f:Y\to\ell_p^s$, then the extension $\hat{f}$ can be done by adding a new coordinate, and defining $\hat{f}:X\to\ell_p^{s+1}$ by setting
for $v\in Y$, $\hat{f}(v)=(f(v),0)$ and $f(x)=(f(u),d(x,u))$. Then we get that for all $v\in Y$,
$\hat{d}(x,v)=\left(\hat{d}(u,v)^p+d(x,u)^p\right)^{1/p}$, which satisfies both requirements for every $1\le p\le\infty$.

\end{itemize}

\section{Terminal Metric Structures for Doubling Metrics}\label{sec:span}

In this section we present our main result.
For ease of notation, we measure the size of the structure as the size per point (e.g. for a spanner with $m$ edges over $n$ points we say the size is $m/n$). Our main result is:

\begin{theorem}\label{thm:main}
Let $(X,d)$ be a metric space with $|X|=n$ that has doubling constant $\lambda$, and fix any set $K\subseteq X$ of size $|K|=k$. For $0<\epsilon <1$, assume that there exists an extendable metric structure for any $Y\subseteq X$ that has stretch $1+\epsilon$ and size $s(|Y|)$, then there exists a structure for $X$ with $1+O(\epsilon)$ stretch for pairs in $K\times X$ and size $s(k\cdot\lambda^{O(\log(1/\epsilon))})+1$.
\end{theorem}

The following corollary follows by applying this theorem with known embeddings/distance oracles/spanners constructions.
\begin{corollary}\label{cor:main}
Let $(X,d)$ be a metric space with $|X|=n$ that has doubling constant $\lambda$, and fix any set $K\subseteq X$ of size $|K|=k$. Then for any $0<\epsilon<1$, the following metric structures exists:
\begin{enumerate}
\item If $(X,d)$ is Euclidean, then there exists a terminal embedding into $\ell_2$ with distortion $1+\epsilon$ and dimension  $O((\log k+\log\lambda\cdot\log(1/\epsilon))/\epsilon^2)$.
\item A terminal embedding into $\ell_\infty$ with distortion $1+\epsilon$ and dimension $\log k\cdot\lambda^{\log(1/\epsilon)+O(1)}\cdot\log(1/\epsilon)$.
\item A terminal spanner for $(X,d)$ with stretch $1+\epsilon$ and $k\cdot\lambda^{O(\log(1/\epsilon))})+n$ edges.

\item A terminal distance oracle with stretch $1+\epsilon$, with size $k\cdot\lambda^{O(\log(1/\epsilon))}+O(n)$ and query time $\lambda^{O(1)}$.
\item A terminal distance labeling scheme with stretch $1+\epsilon$, with label size $\lambda^{O(\log(1/\epsilon))}\cdot\log k\cdot\log\log\Delta_k$ (where $\Delta_k$ is the aspect ratio of $K$).
\item A terminal embedding into a distribution of tree-width $t$ graphs\footnote{See \cite{RS91} for definition of tree-width.} with expected distortion $1+\epsilon$ for $t\le \lambda^{O(\log\log\lambda+\log(1/\epsilon)+\log\log\Delta_K)}$.
\end{enumerate}
\end{corollary}
\begin{proof}
The first item follows from \cite{JL84}, the second using \cite{GKL03,N16}, the third and fourth items use \cite{HM06} results, the fifth applies a result of \cite{S07}, and the sixth from \cite{T04}.\footnote{For the last two results, we note that our proof provides $Y\supseteq K$ satisfying $\Delta_Y\le O(\Delta_K/\epsilon^4)$, on which we apply the labeling scheme of \cite{S07}, or the embedding of \cite{T04}.}
\end{proof}

In what follows we prove \theoremref{thm:main}.
Let $(X,d)$ be a metric space with $|X|=n$ and doubling constant $\lambda$, and let $K\subseteq X$ be a set of terminals.
Fix any $0<\epsilon<1/20$, set $b=\lceil\log(1/\epsilon)\rceil$, and let $\Delta=\max_{u,v\in K}\{d(u,v)\}$, $\delta=\min_{u\neq v\in K}\{d(u,v)\}$ and $s=\lceil\log(\Delta/(\epsilon^2\delta))\rceil$. Let $S=\{0,1,\dots, s\}$, and for each $i\in S$ define $r_i=2^i\cdot \epsilon^2\delta$. Observe that $r_0=\epsilon^2\delta$ and $r_s\ge \Delta$.

\subsection{Construction}\label{sec:construction}

\subsubsection{Multi-Scale Partial Partitions} We begin by constructing partial partitions, based on terminal nets, in various scales. 
The clusters of the partition at level $i$ are created by iteratively taking balls of radius $r_i$ centered at the points of a terminal $r_i$-net. Some of these balls may be sufficiently far away from $K$, we call such clusters {\em final}, and do not partition them in lower levels. See \algref{alg:par-par} for the full details.

\begin{algorithm}[h]
\caption{\texttt{Partial-Partitions $((X,d),K)$}}\label{alg:par-par}
\begin{algorithmic}[1]
\STATE $R_s=X$;
\FOR {$i=s,s-1,\dots,0$}
\STATE Let $N_i=\{x_{i,1},\dots,x_{i,b_i}\}$ be a terminal $r_i$-net of $R_i$; (For $i<s$, each $u\in K\cap N_{i+1}$ will be in $N_i$ as well);
\FOR {$j=1,\dots,b_i$}
\STATE $C_{i,j}\leftarrow B(x_{i,j},r_i)\cap R_i$;
\STATE $R_i\leftarrow R_i\setminus C_{i,j}$;
\IF {$d(x_{i,j},K)\ge r_i/\epsilon$}
\STATE Let $\final(C_{i,j})=true$;
\ELSE
\STATE Let $\final(C_{i,j})=false$
\ENDIF
\ENDFOR
\STATE $R_{i-1}=\bigcup_{j~:~\final(C_{i,j})=false}C_{i,j}$;
\ENDFOR
\end{algorithmic}
\end{algorithm}

For every scale $i\in S$ this indeed forms a partition of $R_i\subseteq X$, because $N_i$ is an $r_i$-net. Also, every cluster $C_{i,j}$ in the partition of $R_i$ has a center $x_{i,j}$.
Observe that every cluster containing a terminal is not final, and that each point in $X$ has at most one final cluster containing it. In addition, the definition of terminal net guarantees that the prefix of $N_i$ consists of terminals, so each terminal $u\in K$ must be assigned to a cluster centered at a terminal. Finally, notice that at level 0, every terminal is a center of its own cluster (since $r_0<\delta$).

\subsubsection{Marking Stage}
We now mark some of the clusters, these marked clusters are the "important" clusters whose center will participate in the black-box construction. For every terminal $u\in K$, let $i_u\in S$ be the maximal index such that $u\in N_{i_u}$, and mark every cluster $C_{i,j}$ with center $x_{i,j}$ satisfying both conditions (recall that $b=\lceil\log(1/\epsilon)\rceil$.)
\begin{enumerate}
\item $i_u-2b\le i\le i_u$, and

\item $d(u,x_{i,j})\le 2r_{i_u}/\epsilon^2$.
\end{enumerate}

\subsubsection{Constructing the Metric Structure}\label{sec:build}
Let $Y\subseteq X$ be the collection of centers of marked clusters (note that $K\subseteq Y$). Apply the black-box construction on $Y$, and extend it to $X\setminus Y$ as follows.
For every $x\in X$ that lies in a final marked cluster $C$ with center $y$, hang $x$ on $y$ (recall that $x$ can be in at most one final cluster). In every other case (e.g., $x$ is in a final unmarked cluster, or does not have a final cluster containing it), hang $x$ on $u\in K$, the nearest terminal to $x$.

\subsection{Analysis}\label{sec:anal}
First we show that $|Y|$ is sufficiently small.
\begin{claim}\label{claim:mark}
$|Y|\le |K|\cdot\lambda^{5b}$.
\end{claim}
\begin{proof}
We will show that each $u\in K$ marks at most $\lambda^{5b}$ clusters. By \claimref{prop:double}, the ball $B(u,r_{i_u+2b+1})$ contains at most $\lambda^{\log(r_{i_u+2b+1}/r_{i_u-2b})}=\lambda^{4b+2}$ net points of $N_{i_u-2b}$ (and only less net points from the other nets $N_{i_u-2b+1},\dots,N_{i_u}$). The second condition for marking implies that only centers in this ball can be marked by $u$. Since there are $2b+1$ possible levels $i\in[i_u-2b,i_u]$, at most $(2b+1)\cdot\lambda^{4b+2}\le\lambda^{5b}$ clusters may be marked by $u$.
\end{proof}

The bound on the size follows from \claimref{claim:mark}, and from the fact that each point in $X\setminus Y$ is hanged from a single $y\in Y$, so it requires a single edge/memory word/coordinate.
It remains to bound the stretch by $1+O(\epsilon)$ for pairs in $K\times X$. By the assumption, the metric structure for $Y$ induces a distance function $\hat{d}$ which is a $1+\epsilon$ approximation of $d$, w.l.o.g we assume that distances cannot contract, and expand by a factor of at most $1+\epsilon$. Fix some $x\in X$ and $v\in K$. Recall that by definition, if $x$ was hanged on $u\in Y$, then $\hat{d}(x,u)$ must satisfy
$$
\max\{d(x,u),\hat{d}(u,v)\}\le\hat{d}(x,u)\le d(x,u)+\hat{d}(u,v)~.
$$
Consider the following cases.

\begin{enumerate}
\item[Case 1:]   $x$ does not have a final cluster containing it. In this case $x$ lies very close to its nearest terminal $u\in K$, and all other terminals are at least $1/\epsilon$ times farther away, so the stretch guaranteed for $u$ will suffice for $x$. More formally: the cluster $C$ containing $x$ at level 0 centered at $y$ is not final, that is, $d(y,K)<r_0/\epsilon$. Since $C$ has radius $r_0=\epsilon^2\delta$, we have that
\begin{equation}\label{eq:err11}
d(x,u)=d(x,K)\le d(x,y)+d(y,K)\le \epsilon^2\delta+\epsilon\delta=(1+\epsilon)\epsilon\delta~.
\end{equation}
We have that $d(u,v)\le d(u,x)+d(x,v)\le (1+\epsilon)\epsilon\delta+d(x,v)\le 2\epsilon\cdot d(u,v)+d(x,v)$,
so that
\begin{equation}\label{eq:err}
d(u,v)\le d(x,v)/(1-2\epsilon)~.
\end{equation}
Since $\hat{d}$ approximates $d$ with stretch $1+\epsilon$ on $K$,
\begin{eqnarray*}
\hat{d}(x,v)&\le&d(x,u)+\hat{d}(u,v)\\
&\le& d(x,u) + (1+\epsilon)d(u,v)\\
&\stackrel{\eqref{eq:err11}}{\le}& (1+\epsilon)\epsilon\delta+(1+\epsilon)d(u,v)\le (1+3\epsilon)d(u,v)\\
&\stackrel{\eqref{eq:err}}{\le}& (1+6\epsilon)d(x,v)~,
\end{eqnarray*}
where the last two inequalities use that $\epsilon<1/12$. On the other hand,
\begin{eqnarray*}
\hat{d}(x,v)&\ge&\hat{d}(u,v)\\&\ge& d(u,v)\\&=& (1-\epsilon)\cdot d(u,v)+\epsilon\cdot d(u,v)\\
&\ge& (1-\epsilon)\cdot (d(x,v)-d(x,u))+\epsilon\delta\\&\stackrel{\eqref{eq:err11}}{\ge}&(1-\epsilon)\cdot d(x,v)-(1-\epsilon)(1+\epsilon)\epsilon\delta+\epsilon\delta\\
&\ge& (1-\epsilon)\cdot d(x,v)~.
\end{eqnarray*}

\item[Case 2:] $x$ lies in a final {\em marked} cluster. Let $C$ be the final marked cluster at level $i\in S$ with center $y$ that contains $x$. In this case we show that $d(x,y)$ is smaller by roughly $1/\epsilon$ than $d(x,K)$, so that the stretch guaranteed for $y\in Y$ will also be sufficient for $x$. Since $C$ is final, $d(y,v)\ge d(y,K)> r_i/\epsilon$, therefore
\begin{equation}\label{eq:ff44}
d(x,v)\ge d(y,v)-d(x,y)\ge r_i/\epsilon-r_i > r_i/(2\epsilon)~.
\end{equation}
Using that the structure built for $Y$ has stretch at most $1+\epsilon$, we have that
\begin{eqnarray*}
\hat{d}(x,v)&\le&d(x,y)+\hat{d}(y,v)\\&\le& d(x,y)+(1+\epsilon)d(y,v)\\
&\le& d(x,y)+(1+\epsilon)(d(x,y)+d(x,v))\\&=& (2+\epsilon)d(x,y)+(1+\epsilon)d(x,v)\\
&\le& (2+\epsilon)r_i +(1+\epsilon)d(x,v)\\
&\stackrel{\eqref{eq:ff44}}{\le}& 2\epsilon(2+\epsilon)d(x,v)+(1+\epsilon)d(x,v)\\
&\le& (1+6\epsilon)d(x,v)~.
\end{eqnarray*}
And also,
\begin{eqnarray*}
\hat{d}(x,v)&\ge&\hat{d}(y,v)\\&\ge& d(y,v)\\&\ge& d(x,v)-d(x,y)\\&\ge& d(x,v)-r_i\\
&\stackrel{\eqref{eq:ff44}}{\ge}& (1-2\epsilon)d(x,v)~.
\end{eqnarray*}

\item[Case 3:] $x$ lies in a final {\em non-marked} cluster $C$. Let $u$ be the nearest terminal to $x$. Intuitively, since $x$ is in a final cluster, all terminals are $1/\epsilon$ farther away than the radius of $C$. However, since $C$ is not marked, its center does not participate in the black-box construction for $Y$. Fortunately, the marking of clusters guarantees that $u$, the closest terminal to $x$, must be in a terminal net of very high scale (otherwise it would have marked $C$), and it follows that every other terminal is either very far away from $u$ (and thus from $x$ as well), or very close to $u$. Surprisingly, in both cases we can use the stretch bound guaranteed for $K$.
We prove this observation formally in the following lemma.
\begin{lemma}\label{lem:single}
For any point $x$ contained in a final non-marked cluster $C$ of level $i$ with $i<s$, there exists a terminal $u'\in K$ such that $d(x,u')\in[r_i/(2\epsilon),3r_i/\epsilon]$ and for any other terminal $w\in K$ it holds that $d(u',w)\le r_i$ or $d(u',w)\ge r_i/\epsilon^2$.
\end{lemma}
\begin{proof}
Since $C$ with center $y$ is the only final cluster containing $x$, the cluster $C'$ with center $y'$ containing $x$ at level $i+1$ is not final (recall we assume $i<s$). Thus there exists a terminal $z\in K$ with $d(y',z)\le r_{i+1}/\epsilon$. Consider the terminal $u'\in N_{i+1}$ which is the center of the cluster containing $z$ at level $i+1$ (we noted above that clusters containing a terminal must have a terminal as a center). By the triangle inequality $d(x,u')\le d(x,y')+d(y',z)+d(z,u')\le r_{i+1}+r_{i+1}/\epsilon+r_{i+1}<3r_i/\epsilon$ (note that the same bound holds for $d(y,u')$). On the other hand, since $C$ is final we have that $d(y,u')\ge r_i/\epsilon$, and thus $d(x,u')\ge d(y,u')-d(y,x)\ge r_i/\epsilon-r_i\ge r_i/(2\epsilon)$.

Next we show that $u'\in N_{i+2b+1}$. Seeking contradiction, assume $u'\notin N_{i+2b+1}$ (or that $i\ge s-2b$ so such a net does not exist), and consider the largest $j$ such that $u'\in N_j$. Since the nets are hierarchical and $u'\in N_{i+1}$, it must be that $i+1\le j\le i+2b$, which implies that $d(u',y)\le 3r_i/\epsilon< r_{i+b+2}< 2r_j/\epsilon^2$. By the marking procedure, the cluster $C$ would have been marked by $u'$. Contradiction. We conclude that $u'\in N_{i+2b+1}$.

Fix any terminal $w\in K$, and we know show that $d(u',w)\le r_i$ or $d(u',w)\ge r_i/\epsilon^2$. Seeking contradiction, assume that $r_i<d(u',w)<r_i/\epsilon^2$. Let $v'\in K$ be the center of the cluster containing $w$ at level $i$, that is $v'\in N_i$. Note that $d(v',w)\le r_i$, and thus $v'\neq u'$. Since $N_{i+2b+1}$ is an $r_{i+2b+1}=2r_i/\epsilon^2$ net, and as $d(u',v')\le r_i+r_{i+2b}$, it must be that $v'\notin N_{i+2b+1}$. The contradiction will follow once we establish that $v'$ will mark $C$. Indeed, the largest $j$ such that $v'\in N_j$ satisfies $i\le j\le i+2b$, and also $d(v',y)\le d(v',w)+d(w,u')+d(u',y)\le r_i+r_{i+2b}+3r_{i+b}\le 2r_j/\epsilon^2$, so $C$ should have been marked.
\end{proof}

Next, we prove the stretch bound for the pair $(x,v)$. Observe that if the final cluster $C$ containing $x$ and centered at $y$ is of level $s$, then $d(y,K)\ge r_s/\epsilon$, and thus
\begin{equation}\label{eq:ggdd}
d(x,K)\ge d(y,K)-d(y,x)\ge r_s/(2\epsilon)~.
\end{equation}
This implies that
\begin{eqnarray*}
\hat{d}(x,v)&\le&d(x,u)+\hat{d}(u,v)\\&\le& d(x,u)+(1+\epsilon)d(u,v)\\
&\le& d(x,v)+ (1+\epsilon)r_s\\&\stackrel{\eqref{eq:ggdd}}{\le}& d(x,v)+ 2\epsilon(1+\epsilon)d(x,v)\\
&\le& (1+3\epsilon)d(x,v)~.
\end{eqnarray*}
Since $d(u,v)\le\Delta\le r_s$, we get that
\begin{eqnarray*}
\hat{d}(x,v)&\ge&d(x,u)\\&\ge&(1-2\epsilon)\cdot(d(x,v)-d(u,v)) + 2\epsilon\cdot d(x,u)\\
&\stackrel{\eqref{eq:ggdd}}{\ge}&(1-2\epsilon)\cdot d(x,v)-r_s+r_s\\&\ge& (1-2\epsilon)\cdot d(x,v)~.
\end{eqnarray*}

From now on we may assume that $C$ is of level $i$ with $i<s$.
By \lemmaref{lem:single} there exists $u'\in K$ such that $d(x,u')\in[r_i/(2\epsilon),3r_i/\epsilon]$ and for any terminal $w\in K$, it holds that $d(u',w)\le r_i$ or $d(u',w)\ge r_i/\epsilon^2$. Note that since $u$ is the nearest terminal to $x$, it must be that $d(u,u')\le r_i$, so we have that $d(x,u)\in[r_i/(3\epsilon),4r_i/\epsilon]$. Finally, we consider the two cases for $v$: close or far from $u'$.\\
{\bf Sub-case a:} $d(u',v)\le r_i$. In this case $d(u,v)\le 2r_i$, and thus $d(x,v)\ge d(x,u)-d(u,v)\ge r_i/(3\epsilon)-2r_i\ge r_i/(4\epsilon)$. It follows that
\begin{eqnarray*}
\hat{d}(x,v)&\le&d(x,u)+\hat{d}(u,v)\\&\le& d(x,u)+(1+\epsilon)d(u,v)\\
&\le& d(x,v) +(1+\epsilon)2r_i\\&\le& d(x,v)+5r_i\\&\le& (1+9\epsilon)d(x,v)~.
\end{eqnarray*}
Since $d(u,v)\le 2r_i\le 8\epsilon\cdot d(x,v)$, we also have
\begin{eqnarray*}
\hat{d}(x,v)&\ge&d(x,u)\\&\ge& d(x,v)-d(u,v)\\&\ge&(1-8\epsilon)\cdot d(x,v)~.
\end{eqnarray*}

{\bf Sub-case b:}  $d(u',v)\ge r_i/\epsilon^2$. Now we have that $d(u',v)\le d(u',x)+d(x,v)\le 3r_i/\epsilon +d(x,v)\le 3\epsilon d(u',v)+d(x,v)$, and so $d(u',v)\le d(x,v)/(1-3\epsilon)$. It follows that
\begin{eqnarray*}
\hat{d}(x,v)&\le&d(x,u)+\hat{d}(u,v)\\&\le& d(x,u)+(1+\epsilon)d(u,v)\\
&\le& (2+\epsilon)d(x,u)+(1+\epsilon)d(x,v)\\&\le& (2+\epsilon)4r_i/\epsilon + (1+\epsilon)d(x,v)\\
&\le& 9\epsilon\cdot d(u',v) +(1+\epsilon)d(x,v)\\
&\le& (1+12\epsilon)d(x,v) ~.
\end{eqnarray*}
Using that $d(u,u')\le r_i$ and that $d(x,v)\ge(1-3\epsilon)d(u',v)\ge (1-3\epsilon)r_i/\epsilon^2\ge r_i/(2\epsilon^2)$, we conclude that
\begin{eqnarray*}
\hat{d}(x,v)&\ge&\hat{d}(u,v)\\&\ge& d(u,v)\ge d(v,x)-d(x,u')-d(u',u)\\
&\ge&d(v,x)-3r_i/\epsilon-r_i\\&\ge& (1-8\epsilon)\cdot d(v,x)+8\epsilon\cdot r_i/(2\epsilon^2)-3r_i/\epsilon-r_i\\
&\ge&(1-8\epsilon)\cdot d(v,x)~.
\end{eqnarray*}
\end{enumerate}

\section{The case where only $K$ is doubling}

So far we assumed that the entire metric $(X,d)$ is doubling. It is quite intriguing to understand what results
can be obtained where only the terminal set $K$ is doubling, while $X$ is arbitrary. We show that in such a case one can obtain terminal metric structures with guarantees similar to the standard results (non-terminal) that apply when the entire metric $(X,d)$ is doubling. For spanners and distance labeling this follow by a simple extension of the black-box result, but unlike \cite{MN07,EFN17}, we use multiple points of $K$ for extending each $x\in X\setminus K$.
\begin{theorem}\label{thm:just-k-doubling}
Let $(X,d)$ be a metric space on $n$ points, and let $K\subseteq X$ so that $(K,d)$ has doubling constant $\lambda$. Then for any $0<\epsilon<1$ there exist:
\begin{itemize}
\item A terminal spanner with stretch $1+\epsilon$ and $O(n\cdot\lambda^{O(\log(1/\epsilon))})$ edges.

\item A terminal distance oracle with stretch $1+\epsilon$, size $n\cdot\lambda^{O(\log(1/\epsilon))}$, and query time $\lambda^{O(1)}$.

\item A terminal labeling scheme with stretch $1+\epsilon$, with label size $\lambda^{O(\log(1/\epsilon))}\log k\cdot\log\log\Delta_k$ (where $\Delta_k$ is the aspect ratio of $K$).
\end{itemize}

\end{theorem}
Observe that the result for the labeling scheme seems to improves \corollaryref{cor:main}, which requires that the whole metric is doubling. (In fact, the label size in \theoremref{thm:just-k-doubling} is slightly larger, this fact is hidden by the constant in the $O(\cdot)$ notation.)

For embeddings, it is unclear how to use this extension approach, since it involves multiple points. We thus need to adjust the embedding itself. As an example to this adjustment, we have the following result, which strictly improves the corresponding item in \corollaryref{cor:main}. Its proof is in \sectionref{app:embed-only-k}.
\begin{theorem}\label{thm:embed}
Let $(X,d)$ be a metric space, and let $K\subseteq X$ of size $|K|=k$ so that $(K,d)$ has doubling constant $\lambda$. Then for any $0<\epsilon<1$ there exists a terminal embedding of $X$ into $\ell_\infty$ with distortion $1+\epsilon$, and dimension $\log k\cdot\lambda^{O(\log(1/\epsilon))}$.
\end{theorem}
We remark that any embedding of $(X,d)$ into $\ell_\infty$ with distortion less than 3 for all pairs, requires in general dimension $\Omega(n)$ \cite{matbook}. We also note that a terminal version of the JL lemma is impossible whenever only $K$ is Euclidean, and $X\setminus K$ is not. To see this, note that any three vertices of $K_{2,2}$ admit an isometric embedding to $\ell_2$, but embedding all four requires distortion $\sqrt{2}$. When only one vertex is non-terminal, all pairwise distances must be preserved up to $1+\epsilon$, which is impossible for $\epsilon<1/3$, say.

\subsection{Proof of \theoremref{thm:just-k-doubling}}

We prove the spanner result first.
Let $H$ be a spanner for $(K,d)$ with stretch $1+\epsilon$ and $k\cdot\lambda^{O(\log(1/\epsilon))}$ edges given by \cite{HM06}, say. For any $x\in X$, let $u=u(x)\in K$ be the closest terminal to $x$, and denote $R=d(x,u)$. Take $N(x)$ to be an $\epsilon R$-net of $B(x,2R/\epsilon)\cap K$, by \claimref{prop:double}, $|N(x)|\le \lambda^{O(\log(1/\epsilon))}$. Add the edges $\{(x,v)\}_{v\in N(x)}$, each with weight $d(x,v)$ to the spanner. Since we added $\lambda^{O(\log(1/\epsilon))}$ edges for each point, the bound on the number of edges follows, and it remains to bound the stretch by $1+O(\epsilon)$. Clearly no distances can contract, and we bound the expansion. Fix $x\in X$ and $v\in K$, and denote $u=u(x)$ with $R=d(x,u)$. In the case $v\notin B(x,2R/\epsilon)$ we have that $R\le \epsilon\cdot d(x,v)/2$, so that
\begin{eqnarray*}
d_H(x,v)&\le& d_H(x,u)+d_H(u,v)\le d(x,u)+(1+\epsilon)d(u,v)\\&\le&(2+\epsilon)d(x,u)+(1+\epsilon)d(x,v)= (2+\epsilon)R+(1+\epsilon)d(x,v)\\&\le& (1+3\epsilon)d(x,v)~.
\end{eqnarray*}
Otherwise, $v\in B(x,2R/\epsilon)$. Let $v'\in N(x)$ be the nearest net point to $v$, with $d(v,v')\le\epsilon R\le \epsilon\cdot d(x,v)$ (recall $u$ is the nearest terminal to $x$). Then
\begin{eqnarray*}
d_H(x,v)&\le& d_H(x,v')+d_H(v',v)\\&\le& d(x,v')+(1+\epsilon)d(v',v)\\&\le& d(x,v)+(2+\epsilon)d(v',v)\\&\le& d(x,v)+(2+\epsilon)\epsilon\cdot d(x,v)\\&\le& (1+3\epsilon)d(x,v)~.
\end{eqnarray*}

The proof for the labeling scheme (and also distance oracle) is similar. Apply the black-box scheme on $(K,d)$, and for each $x\in X\setminus K$ define $N(x)$ as above, and $x$ stores all labels for $v'\in N(x)$ along with $d(x,v')$. Given a query $(x,v)$, return $\min_{v'\in N(x)}\{d(x,v')+\hat{d}(v,v')\}$, where $\hat{d}$ is the distance function of the labeling scheme.

\subsubsection{Lower Bound}
We now show that when only $K$ is doubling, one cannot achieve a result as strong as \theoremref{thm:main} (there the number of edges in a spanner with stretch $1+\epsilon$ can be as low as $n+o(n)$). In fact, \theoremref{thm:just-k-doubling} is tight up to a constant factor in the exponent of $\lambda$. 
\begin{claim}\label{claim:lower}
There exists a constant $c>0$, so that for any (sufficiently large) integer $n$ and any integer $\lambda>1$, there is a metric $(X,d)$ on $n$ points with a subset $K\subseteq X$, so that $(K,d)$ has doubling constant $O(\lambda)$, but for any $0<\epsilon<1$, any terminal spanner of $X$ with stretch $1+\epsilon$ must have at least $n\cdot\lambda^{\log(c/\epsilon)}$ edges.
\end{claim}
\begin{proof}
Let $t=\lceil\log\lambda\rceil$, and let $K$ be an $\epsilon$-net of the unit sphere of $\R^t$. It is well known that $|K|=\Theta(1/\epsilon)^{t-1}=\lambda^{\log(c/\epsilon)}$ for some constant $c$.\\ Define $(X,d)$ by setting for each $x,y\in X$, $d(x,y)=\left\{\begin{array}{ccc} \|x-y\|_2 & x,y\in K\\ 1& x\in X\setminus K, y\in K\\
2 & x,y\in X\setminus K \end{array} \right.$.
Note that distances between points in $K$ correspond to the Euclidean distance, and are at most 2, so that $K$ has doubling constant $O(\lambda)$. Observe that any spanner with stretch $1+\epsilon$ must contain all the edges in $K\times X$, because the distance between any two points in $K$ is larger than $\epsilon$, so any path from $x\in X\setminus K$ to $y\in K$ that does not contain the edge $(x,y)$, will be of length greater than $1+\epsilon$.
\end{proof}



\subsection{Proof of \theoremref{thm:embed}}\label{app:embed-only-k}

We follow the embedding technique of \cite{N16}, but with different edge contractions defined below. Assume w.l.o.g that the minimal distance in $(X,d)$ is 1.
Let $\Delta = \diam(X)$, and for all $0\le i\le \log\Delta$ let $(X,d_i)$ be the metric defined as follows: consider the complete graph on vertex set $X$, with edge $\{u,v\}$ having weight $d(u,v)$. For every $x\in X$ and $v\in K$ with $d(x,v)<2^{i-1}\cdot \epsilon/k$, replace the weight of this edge by $0$, and let $d_i$ be the shortest path metric on this graph. Since any shortest path in this graph has at most $2k$ edges that contain a vertex in $K$, we have that $d(x,y)-\epsilon\cdot 2^i\le d_i(x,y)\le d(x,y)$ for all $x,y\in X$.

For each $0\le i\le \log\Delta$ take a $r_i$-net $N_i$ with respect to $(K,d_i)$ (i.e., take only terminals to the net), where $r_i=\epsilon\cdot 2^{i-2}$. Partition each $N_i$ into $t=\lambda^{O(\log(1/\epsilon))}$ sets $N_{i1},\dots,N_{it}$, such that for each $u,v\in N_{ij}$, $d_i(u,v)\ge 5\cdot 2^i$. (To obtain $N_{ij}$, one can greedily choose points from $N_i\setminus(\bigcup_{j'<j}N_{ij'})$ until no more can be chosen. See \cite{N16} for details.)
Next we define the embedding, fix $D=\lceil 2t\log (2k/\epsilon)\rceil$, and let $\{e_0,\dots,e_{D-1}\}$ be the standard orthonormal basis for $\R^D$, extended to an infinite sequence $\{e_j\}_{j\in\N}$ (that is, $e_j = e_{j~(\textrm{mod } D)}$ for all $j\in \N$).
For any $0\le i\le\log\Delta$ and $0\le j\le t-1$, for $x\in X$ let
\[
g_{ij}(x)=\min\{2^{i+1}, d_i(x,N_{ij})\}~.
\]
Define the embedding $f:X\to\R^D$ by
\[
f(x) = \sum_{i=0}^{\log\Delta}\sum_{j=0}^{t-1}g_{ij}(x)\cdot e_{it+j}~.
\]

{\bf Expansion Bound:}
Now we show that the embedding $f$ under the $\ell_\infty$ norm does not expand distances for pairs in $X\times K$ by more than a factor of $1+\epsilon$. Fix a pair $x\in X$ and $v\in K$, and consider the $h$-th coordinate of the embedding $f_h$, with $0\le h\le D-1$. We have that $f_h(x)-f_h(v)=\sum_{i,j~:~ h=it+j(\mmod~D)}g_{ij}(x)-g_{ij}(v)$. Let $0\le i'\le\log\Delta$ be such that $2^{i'-1}\le d(x,v)<2^{i'}$, then for all $i>i'+\log (2k/\epsilon)$ it holds that $d(x,v)<2^{i-1}\cdot\epsilon/k$ and thus $d_i(x,v)=0$, in particular, $g_{ij}(x)=g_{ij}(v)$ and so there is no contribution at all from such scales. By the triangle inequality we also have that $g_{ij}(x)-g_{ij}(v)\le d_i(x,v)$ and $g_{ij}(x)-g_{ij}(v)\le 2^{i+1}$ for all $0\le i\le \log\Delta$ and $0\le j\le t-1$.
\begin{eqnarray*}
f_h(x)-f_h(v)&\le&\sum_{i,j~:~ i\le i'+\log (2k/\epsilon), h=it+j(\mmod~D)}g_{ij}(x)-g_{ij}(v)\\
&\le&\sum_{i,j~:~ i'-\log (2k/\epsilon)< i \le i'+\log (2k/\epsilon), h=it+j(\mmod~D)}g_{ij}(x)-g_{ij}(v)\\
 &&+\sum_{i\le i'-\log (2k/\epsilon)}2^{i+1}\\
&\le&d_i(x,v)+2^{i'+1}\cdot\epsilon/k\\
&\le& d(x,v)(1+\epsilon)~.
\end{eqnarray*}
The third inequality holds, since by the choice of $D$ there is at most one possible choice of $i,j$ with $i'-\log (2k/\epsilon)<i<i'+\log (2k/\epsilon)$ such that $h=it+j(\mmod ~D)$, and the last inequality uses that $k\ge 4$. By symmetry it follows that $|f_h(x)-f_h(v)|\le d(x,v)(1+\epsilon)$, and thus $|f(x)-f(v)|\le d(x,v)(1+\epsilon)$.

{\bf Contraction Bound:}
Now we bound the contraction of the embedding for pairs containing a terminal. Fix $x\in X$ and $v\in K$. We will show that there exists a single coordinate $0\le h\le D-1$ such that $|f_h(x)-f_h(v)|\ge (1-\epsilon)d(x,v)$. Let $0\le i\le\log\Delta$ such that $2^i\le d(x,v)<2^{i+1}$, and let $0\le j\le t-1$ be such that $d_i(v,N_{ij})\le r_i$ (such a $j$ must exist because $N_i$ is an $r_i$-net of $K$). Denote by $u\in N_{ij}$ the point satisfying $d_i(v,N_{ij})=d_i(v,u)$. Since $r_i=\epsilon\cdot 2^{i-2}$ also $g_{ij}(v)\le r_i$.

We claim that $d_i(x,N_{ij})=d_i(x,u)$. To see this, first observe that $d_i(x,u)\le d_i(x,v)+d_i(v,u)\le 2^{i+1}+r_i<(5/4)\cdot 2^{i+1}$. Consider any other $y\in N_{ij}$, by the construction of $N_{ij}$, $d_i(y,u)\ge 5\cdot 2^i$, so $d_i(y,x)\ge d_i(y,u)-d_i(x,u) > (5/2)\cdot 2^{i+1}-(5/4)\cdot 2^{i+1} = (5/4)\cdot 2^{i+1} > d_i(x,u)$. Thus it follows that either $g_{ij}(x)=2^{i+1}\ge d_i(x,y)$, or $g_{ij}(x)=d_i(x,u)\ge d_i(x,v)-d_i(v,u)\ge d_i(x,v) - r_i$.
Using that $d_i(x,v)\ge d(x,v)-\epsilon\cdot 2^i$, we conclude that
\[
g_{ij}(x) - g_{ij}(v)\ge (d_i(x,v) - r_i) - r_i = d_i(x,v)-\epsilon\cdot 2^{i-1} \ge d(x,v)-2\epsilon\cdot 2^i\ge(1-2\epsilon)\cdot d(x,v)~.
\]

Let $0\le h\le D-1$ be such that $h=it+j(\mmod~D)$, for the values of $i,j$ fixed above. Then we claim that any other pair $i',j$ such that $h=i'k+j(\mmod~D)$ has either $0$ or very small contribution to the $h$ coordinate. If $i'>i$ then it must be that $i'\ge \log (2k/\epsilon) + i+1$ so that $d(x,v)\le 2^{i+1}<2^{i'-1}\cdot\epsilon/k$, thus as before $g_{i'j}(x)=g_{i'j}(v)$. For values of $i'$ such that $i'<i$, then $i'\le i-\log (2k/\epsilon)$, thus
\begin{eqnarray*}
\sum_{i'<i,j ~:~ h=i't+j (\mmod~D)}|g_{i'j}(x) - g_{i'j}(v)|&\le&\sum_{i'\le i- \log (2k/\epsilon)}2^{i'+1}\\
&\le& 2^i\cdot 2\epsilon/k\\
&\le& \epsilon\cdot d(x,v)~.
\end{eqnarray*}
Finally,
\begin{eqnarray*}
\|f(x)-f(v)\|_\infty&\ge& |f_h(x)-f_h(v)|\\
&\ge& |g_{ij}(v) - g_{ij}(x)|-\!\!\!\!\!\!\!\!\sum_{i'<i,j ~:~ h=i't+j (\mmod~D)}\!\!\!|g_{i'j}(x) - g_{i'j}(v)|\\
&\ge&d(x,v)(1-3\epsilon)~.
\end{eqnarray*}

\section{Acknowledgements}
We are grateful to Paz Carmi for fruitful discussions.

\bibliographystyle{alpha}
\bibliography{spanner}

\newcommand{\etalchar}[1]{$^{#1}$}
\begin{thebibliography}{ADD{\etalchar{+}}93}

\bibitem[AB18]{AB18}
Amir Abboud and Greg Bodwin.
\newblock Reachability preservers: New extremal bounds and approximation
  algorithms.
\newblock In {\em Proceedings of the Twenty-Ninth Annual {ACM-SIAM} Symposium
  on Discrete Algorithms, {SODA} 2018, New Orleans, LA, USA, January 7-10,
  2018}, pages 1865--1883, 2018.

\bibitem[ADD{\etalchar{+}}93]{ADDJS93}
I.~Alth\"ofer, G.~Das, D.~Dobkin, D.~Joseph, and J.~Soares.
\newblock On sparse spanners of weighted graphs.
\newblock {\em Discrete Comput. Geom.}, 9:81--100, 1993.

\bibitem[Ass83]{A83}
P.~Assouad.
\newblock Plongements lipschitziens dans $\mathbb{R}^n$.
\newblock {\em Bull. Soc. Math. France}, 111(4):429--448, 1983.

\bibitem[Bar96]{B96}
Y.~Bartal.
\newblock Probabilistic approximation of metric spaces and its algorithmic
  applications.
\newblock In {\em Proceedings of the 37th IEEE Symp. on Foundations of Computer
  Science}, pages 184-- 193, 1996.

\bibitem[Bou85]{B85}
J.~Bourgain.
\newblock On lipschitz embedding of finite metric spaces in hilbert space.
\newblock {\em Israel Journal of Mathematics}, 52(1-2):46--52, 1985.

\bibitem[CDNS92]{CDNS92}
Barun Chandra, Gautam Das, Giri Narasimhan, and Jos{\'e} Soares.
\newblock New sparseness results on graph spanners.
\newblock In {\em Proc. of 8th SOCG}, pages 192--201, 1992.

\bibitem[CE05]{CE05}
D.~Coppersmith and M.~Elkin.
\newblock Sparse source-wise and pair-wise distance preservers.
\newblock In {\em SODA: ACM-SIAM Symposium on Discrete Algorithms}, pages
  660--669, 2005.

\bibitem[CG06]{CG06}
T-H.~Hubert Chan and Anupam Gupta.
\newblock Small hop-diameter sparse spanners for doubling metrics.
\newblock In {\em Proceedings of the Seventeenth Annual ACM-SIAM Symposium on
  Discrete Algorithm}, SODA '06, pages 70--78, Philadelphia, PA, USA, 2006.
  Society for Industrial and Applied Mathematics.

\bibitem[CGK13]{CGK13}
Marek Cygan, Fabrizio Grandoni, and Telikepalli Kavitha.
\newblock On pairwise spanners.
\newblock In {\em 30th International Symposium on Theoretical Aspects of
  Computer Science, {STACS} 2013, February 27 - March 2, 2013, Kiel, Germany},
  pages 209--220, 2013.

\bibitem[CGMZ16]{CGMZ06}
T.-H.~Hubert Chan, Anupam Gupta, Bruce~M. Maggs, and Shuheng Zhou.
\newblock On hierarchical routing in doubling metrics.
\newblock {\em ACM Trans. Algorithms}, 12(4):55:1--55:22, August 2016.

\bibitem[CLNS15]{CLNS15}
T.{-}H.~Hubert Chan, Mingfei Li, Li~Ning, and Shay Solomon.
\newblock New doubling spanners: Better and simpler.
\newblock {\em {SIAM} J. Comput.}, 44(1):37--53, 2015.

\bibitem[DHN93]{DHN93}
Gautam Das, Paul~J. Heffernan, and Giri Narasimhan.
\newblock Optimally sparse spanners in 3-dimensional euclidean space.
\newblock In {\em Proceedings of the Ninth Annual Symposium on Computational
  GeometrySan Diego, CA, USA, May 19-21, 1993}, pages 53--62, 1993.

\bibitem[EFN15]{EFN15}
Michael Elkin, Arnold Filtser, and Ofer Neiman.
\newblock Prioritized metric structures and embedding.
\newblock In {\em Proceedings of the Forty-Seventh Annual {ACM} on Symposium on
  Theory of Computing, {STOC} 2015, Portland, OR, USA, June 14-17, 2015}, pages
  489--498, 2015.

\bibitem[EFN17]{EFN17}
Michael Elkin, Arnold Filtser, and Ofer Neiman.
\newblock Terminal embeddings.
\newblock {\em Theor. Comput. Sci.}, 697:1--36, 2017.

\bibitem[ES15]{ES15}
Michael Elkin and Shay Solomon.
\newblock Optimal euclidean spanners: Really short, thin, and lanky.
\newblock {\em J. {ACM}}, 62(5):35:1--35:45, 2015.

\bibitem[GGN06]{GGN06}
Jie Gao, Leonidas~J. Guibas, and An~Nguyen.
\newblock Deformable spanners and applications.
\newblock {\em Comput. Geom. Theory Appl.}, 35(1-2):2--19, August 2006.

\bibitem[GKL03]{GKL03}
Anupam Gupta, Robert Krauthgamer, and James~R. Lee.
\newblock Bounded geometries, fractals, and low-distortion embeddings.
\newblock In {\em Proceedings of the 44th Annual IEEE Symposium on Foundations
  of Computer Science}, FOCS '03, pages 534--, Washington, DC, USA, 2003. IEEE
  Computer Society.

\bibitem[Got15]{G15}
Lee{-}Ad Gottlieb.
\newblock A light metric spanner.
\newblock In {\em Proc. of 56th FOCS}, pages 759--772, 2015.

\bibitem[GR08]{GR08}
Lee{-}Ad Gottlieb and Liam Roditty.
\newblock An optimal dynamic spanner for doubling metric spaces.
\newblock In {\em Algorithms - {ESA} 2008, 16th Annual European Symposium,
  Karlsruhe, Germany, September 15-17, 2008. Proceedings}, pages 478--489,
  2008.

\bibitem[HPM06]{HM06}
Sariel Har-Peled and Manor Mendel.
\newblock Fast construction of nets in low-dimensional metrics and their
  applications.
\newblock {\em SIAM J. Comput.}, 35(5):1148--1184, May 2006.

\bibitem[JL84]{JL84}
William Johnson and Joram Lindenstrauss.
\newblock Extensions of {L}ipschitz mappings into a {H}ilbert space.
\newblock In {\em Conference in modern analysis and probability (New Haven,
  Conn., 1982)}, volume~26 of {\em Contemporary Mathematics}, pages 189--206.
  American Mathematical Society, 1984.

\bibitem[Kav15]{K15}
Telikepalli Kavitha.
\newblock New pairwise spanners.
\newblock In {\em 32nd International Symposium on Theoretical Aspects of
  Computer Science, {STACS} 2015, March 4-7, 2015, Garching, Germany}, pages
  513--526, 2015.

\bibitem[LLR95]{LLR95}
N.~Linial, E.~London, and Y.~Rabinovich.
\newblock The geometry of graphs and some of its algorithmic applications.
\newblock {\em Combinatorica}, 15(2):215--245, 1995.

\bibitem[Mat96]{M96}
J.~Matou{\v{s}}ek.
\newblock On the distortion required for embeding finite metric spaces into
  normed spaces.
\newblock volume~93, pages 333--344, 1996.

\bibitem[Mat02]{matbook}
Jiri Matousek.
\newblock {\em Lectures on Discrete Geometry}.
\newblock Springer-Verlag New York, Inc., Secaucus, NJ, USA, 2002.

\bibitem[MN07]{MN07}
Manor Mendel and Assaf Naor.
\newblock Ramsey partitions and proximity data structures.
\newblock {\em Journal of the European Mathematical Society}, 9(2):253--275,
  2007.

\bibitem[Nei16]{N16}
Ofer Neiman.
\newblock Low dimensional embeddings of doubling metrics.
\newblock {\em Theory Comput. Syst.}, 58(1):133--152, 2016.

\bibitem[NS07]{NS07}
Giri Narasimhan and Michiel Smid.
\newblock {\em Geometric Spanner Networks}.
\newblock Cambridge University Press, New York, NY, USA, 2007.

\bibitem[Par14]{P14}
Merav Parter.
\newblock Bypassing erd{\H{o}}s' girth conjecture: Hybrid stretch and
  sourcewise spanners.
\newblock In {\em Automata, Languages, and Programming - 41st International
  Colloquium, {ICALP} 2014, Copenhagen, Denmark, July 8-11, 2014, Proceedings,
  Part {II}}, pages 608--619, 2014.

\bibitem[RS91]{RS91}
Neil Robertson and P.~D. Seymour.
\newblock Graph minors: X. obstructions to tree-decomposition.
\newblock {\em J. Comb. Theory Ser. B}, 52(2):153--190, June 1991.

\bibitem[Sli07]{S07}
Aleksandrs Slivkins.
\newblock Distance estimation and object location via rings of neighbors.
\newblock {\em Distributed Computing}, 19(4):313--333, 2007.

\bibitem[Tal04]{T04}
Kunal Talwar.
\newblock Bypassing the embedding: Algorithms for low dimensional metrics.
\newblock In {\em Proceedings of the Thirty-sixth Annual ACM Symposium on
  Theory of Computing}, STOC '04, pages 281--290, New York, NY, USA, 2004. ACM.

\end{thebibliography}

\end{document}